\documentclass[12pt]{article}
\usepackage[T2A]{fontenc}
\usepackage[utf8]{inputenc}
\usepackage[russian,english]{babel}
\usepackage{amsmath,amsfonts,amssymb,mathrsfs,amscd,amsthm}
\usepackage{cite}
\usepackage{graphicx}
\usepackage{diagbox}
\usepackage{hyperref}
\hypersetup{
colorlinks=true,
linkcolor=red,
citecolor=blue,
urlcolor=blue
}
\usepackage{appendix}
\usepackage{subfigure} 

\allowdisplaybreaks
\usepackage{multirow}

\usepackage[a4paper, margin = 1in]{geometry}

\newtheorem{theorem}{Theorem}
\newtheorem{proposition}{Proposition}
\newtheorem{remark}{Remark}
\newtheorem{definition}{Definition}

\usepackage{xcolor} 

\title{\bf  L\'{e}vy Laplacian on manifold and heat flows of differential forms}

\author{Boris Volkov\footnote{E-mail: \url{borisvolkov1986@gmail.com}; \href{https://www.mathnet.ru/eng/person94935}{mathnet.ru/eng/person94935}
}
\\
\vspace{0.2cm} \\
Steklov Mathematical Institute of Russian Academy of Sciences, \\
8 Gubkina str., Moscow, 119991, Russia}

\date{}

\begin{document}
\maketitle

\begin{abstract} 
The L\'evy Laplacian is an infinite-dimensional differential operator, which is interesting for its connection with the Yang-Mills gauge fields. The article proves the equivalence of various definitions of the L\'evy Laplacian on the manifold of $H^1$-paths on a Riemannian manifold.
The heat equation with the L\'evy Laplacian is considered. The tendency of some solutions of this heat equation to the locally constant functionals as time tends to infinity is studied. These solutions are constructed using heat flows of differential forms on the compact Riemannian manifold.
\end{abstract}


\section{Introduction}

The original definition of the L\'evy Laplacian by Paul L\'evy is the following~\cite{L1951}.
Let $\{e_n\}$ be an orthonormal basis in $L_2([0,1],\mathbb{R})$.  Then the L\'evy Laplacian (generated by the orthonormal basis $\{e_n\}$)
acts on a twice Fr\'{e}chet differentiable real-valued function $f$ on $L_2([0,1],\mathbb{R})$ as the Ces\`aro  mean of the second order directional derivatives along the vectors from this basis,
\begin{equation}
\label{Levy2}
\Delta^{\{e_n\}}_L f(x)=\lim_{n\to \infty} \frac 1n \sum_{k=1}^n\langle f''(x)e_k,e_k\rangle.
\end{equation}
Another way to define the L\'evy Laplacian is as an integral functional generated by a special type of the second derivative. Under certain conditions on an orthonormal basis $\{e_n\}$ these definitions coincide.
By analogy with these definitions, but in a more complex way, the Laplacian, which is also called the L\'evy Laplacian, can be defined for functions on an infinite-dimensional manifold of paths~\cite{AGV1993,AGV1994,LV2001}. The study of the L\'evy Laplacian on a manifold is motivated by its connection with the Yang-Mills gauge fields.
The equivalence of the Yang--Mills equations and the Laplace equations for the L\'evy Laplacian for connection in the trivial vector bundle over $\mathbb{R}^d$~\cite{AGV1993,AGV1994}
and in a vector bundle over the Riemannian manifold~\cite{LV2001} was shown. Namely, it was proved that a connection in a vector bundle is a solution of the Yang---Mills equations if and only if the parallel transport generated by the connection is a solution of the Laplace equation for the L\'evy Laplacian. Lately, it was shown that the L\'evy Laplacian is not invariant under infinite-dimensional rotations.  In the case of a path manifold on a four-dimensional Riemannian manifold, one can obtain the so-called modified L\'evy Laplacian  by the action of a special infinie dimensional rotation on the L\'evy Laplacian. A connection in a vector bundle over the  four-dimensional Riemannian manifold is an instanton if and only if the parallel transport generated by the connection is a solution of the Laplace equation for the modified L\'evy Laplacian~\cite{VolkovLLI,Volkov2020,Volkov2023,Volkov2022}. The  L\'evy Laplacian in the Malliavin calculus and its connection with the Yang--Mills theory were studied in~\cite{LV2001,Volkov2017,Volkov2018}.

 The approach to the Yang--Mills fields based on the L\'evy Laplacian goes back to the work on the connection between the Yang--Mills equations and infinite-dimensional loop equations~\cite{Polyakov1979,Aref'eva1979,Aref'eva1980,Polyakov1980,AV1981, Polyakov1987}. In particular, in the work~\cite{AV1981} the equations of motion of chiral fields on parallel transport with divergence associated with the L\'evy Laplacian were considered (see also~\cite{Volkov2017,Volkov2019}). Various approaches to the Yang-Mills equations based on parallel transport and holonomy, but not using the L\'evy Laplacian, were studied in~\cite{Gross,Driver,Bauer1998,Bauer2003,ABT,AT}. A review of homotopy and  holonomy approaches to gauge theories can be found in~\cite{Meneses}. In this regard the following result is important. If an operator-valued function on a path manifold has certain properties (strongly differentiability, group property, invariance with respect to reparameterization), then it is a parallel transport generated by some connection~\cite{Driver}. This result can be interpreted using the theory of induced representations for groupoids~\cite{Gibilisco1997}.

One of the main results of this paper is a proof of the equivalence of three definitions of the L\'evy Laplacian on an infinite-dimensional manifold: the covariant definition of the L\'evy Laplacian as a composition of the special gradient and the L\'evy divergence~\cite{Volkov2019a}, the original definition based on a parallelization  of the infinite dimensional manifold~\cite{LV2001}, and the definition of the L\'evy Laplacian as the Ces\`aro mean of the second directional derivatives~\cite{AS2006}. The second definition is a natural generalization of the definition for the flat case from~\cite{AGV1994}. The relationship between the different definitions of the L\'evy Laplacian has not been previously investigated  for the case of the manifold but only for the flat case in~\cite{Volkov2019}.

Using integrals of functions and 1-forms along curves, a certain class of functionals on the path space is constructed on which the action of the L\'evy Laplacian is well defined.
This class of functionals includes a family of eigenvectors for the L\'evy Laplacian. 
We study the heat equation for the L\'evy Laplacian.
We show that the heat flows of differential 0-forms and 1-forms generate solutions of the heat equation with the L\'evy Laplacian. Using Milgram---Rosembloom theorem~\cite{MilgramRosenbloom}, the pointwise convergence of solutions of this heat equation to locally constant functionals as time tends to infinity is shown. These locally constant functionals  take the same value on homotopy equivalent curves. The study of the heat equation for the L\'evy Laplacian and the behavior of its solution as time tends to infinity is of interest in particular in connection with the Yang-Mills heat equation~\cite{Volkov2019a}. In the commutative case, the Yang-Mills heat equations reduce to the heat equation for 1-forms~\cite{Rade}.

The paper is organized as follows. In~Sec.~\ref{SecManifolds},
we give
preliminary information about vector bundles over manifolds of curves.
In Sec.~\ref{SecCovLevyLaplacian},
 we present the definition of the L\'evy Laplacian as the composition
of the special divergence and the
gradient and discuss some properties of these differential operators. In Sec.~\ref{Sec4}
and in Sec.~\ref{CesaroMean}, we discuss different definitions of the L\'evy Laplacian on the manifold and prove the equivalence of three definitions of this operator. 
In Sec.~\ref{Secfunctionals}, we present the action of  the L\'evy Laplacian on the integrals along the curves and construct some families of eigenvectors for this operator.   In Sec.~\ref{SecHeat}, we show that the heat flows of differential 0-forms and 1-forms generate solutions of the heat equation with the L\'evy Laplacian. We also show the tendency of these solutions to the locally constant functionals as time tends to infinity is studied.
In Sec.~\ref{SecYMH}, the relationship between the  heat equation with the L\'evy Laplacian and the Yang-Mills heat equations is discussed.
Conclusions Sec.~\ref{Sec:Conclusions} resumes the work.

\section{Manifold of $H^1$-curves}
\label{SecManifolds}

In this section, we give preliminaries about infinite-dimensional manifolds of paths.
For more information, see~\cite{Driver,Klingenberg,Klingenberg2}.

The symbols
$H^0=H^0([0,1],\mathbb{R}^d)$  and $H^1=H^1([0,1],\mathbb{R}^d)$   denote the Hilbert spaces of $L_2$-functions
and $H^1$-functions  on
$[0,1]$ with values in $\mathbb{R}^d$ respectively.  These spaces are  endowed with standard scalar products
$(h_1,h_2)_0=\int_0^1(h_1(\tau),h_2(\tau))_{\mathbb{R}^d}d\tau$
and $(h_1,h_2)_1=(h_1,h_2)_0+(\dot{h}_1,\dot{h}_2)_0$
   respectively.
Let $H_0^1=\{h\in H^1\colon h(0)=0\}$
and $H_{0,0}^1=\{h\in H_0^1\colon h(1)=0\}$.

Let $M$ be a $d$-dimensional connected  Riemannian manifold and let $g$ denote the Riemannian metric on $M$. Let $\mathcal{P}ath$ be the Hilbert manifold of $H^1$-curves in  $M$.  This Hilbert manifold 
is modeled over the Hilbert space $H^1$~\cite{Driver}. We consider vector bundles $\mathcal H^0=\mathcal H^0(\mathcal{P}ath)$ and $\mathcal H^1=\mathcal H^1(\mathcal{P}ath)$   over the Hilbert manifold $\mathcal{P}ath$ of $H^0$ vector fields and $H^1$ vector fields along curves from $\mathcal{P}ath$ respectively~\cite{Klingenberg2,Klingenberg}. Both vector bundles will be required bellow for the covariant definition of the L\'evy Laplacian.
The fiber of $\mathcal H^0$ over $\gamma\in\mathcal{P}ath$ is denoted by $H^0_\gamma(TM)$. The Riemannian metric on the bundle $\mathcal H^0$ 
has the form
$$
G_{0}(X(\gamma),Y(\gamma))=\int_0^1g(X(\gamma;\tau),Y(\gamma;\tau))d\tau,
$$
where $X(\gamma)=X(\gamma;\cdot)$, $Y(\gamma)=Y(\gamma;\cdot)$ are vector fields along $\gamma$.

Let $H^1_\gamma(TM)$ denote a fiber of $\mathcal H^1$ over $\gamma$.
A $H^1$ vector field along $\gamma\in \mathcal{P}ath$ is a section in 
the pullback bundle $\gamma^\ast TM$.
 The Levi-Civita connection induces the connection on $\gamma^\ast TM$ and 
the covariant derivative $\nabla X(\gamma)$ of $X(\gamma)\in H^1_\gamma(TM)$  is the $H^0$ vector field along $\gamma$  such that
\begin{equation}
\label{nablaX}
\nabla X(\gamma;\tau)= \dot{X}(\gamma;\tau)+\Gamma(\gamma(\tau))(X(\gamma;\tau),\dot{\gamma}(\tau)).
\end{equation}
Here $\Gamma=(\Gamma^\mu_{\lambda \nu})$ are Cristoffel symbols of the Levi-Civita connection on $M$.  
The Riemannian metric 
on the bundle $\mathcal H^1$  is defined by the formula
$$
G_1(X(\gamma),Y(\gamma))=G_0(X(\gamma),Y(\gamma))+G_0(\nabla X(\gamma),\nabla Y(\gamma)).
$$
Let $Q_{\tau_2,\tau_1}(\gamma)\colon T_{\gamma(\tau_1)}M\to T_{\gamma(\tau_2)}M$ denote the parallel transport generated by the Levi-Civita connection  along the restriction of the curve  $\gamma\in \mathcal{P}ath$ on the interval $[\tau_1,\tau_2]$. Formula~(\ref{nablaX}) implies that for a $H^1$ vector field $X$ along $\gamma\in \mathcal{P}ath_m$ holds
\begin{equation}
\label{Xgamma}
\nabla{X}(\gamma;\tau)=Q_{\tau,0}(\gamma)\frac d{d\tau}(Q_{\tau,0}(\gamma)^{-1}X(\gamma;\tau)).
\end{equation}
So, the parallel transport $Q$ defines an isometric isomorphism between $H^1_\gamma(TM)$ and $H^1$. 
Let $\mathcal{H}^1_{0,0}$
denote the subbundle of $\mathcal{H}^1$
such that its fiber over $\gamma \in \mathcal{P}ath$ is a Hilbert space ${(H^1_{0,0})}_\gamma(TM)$  of all fields $X$ from $H^1_{\gamma}(TM)$ such that $X(\gamma;0)=X(\gamma;1)$.

The Levi-Civita connection on the $d$-dimensional manifold $M$
generates the canonical connection $\nabla^{\mathcal{H}^0}$ in the
infinite-dimensional bundle $\mathcal {H}^0$~\cite{Klingenberg2,Klingenberg}.
Let  $d(\cdot,\cdot)$ denote the distance on $M$  generated by the Riemannian  metric   $g$. For $\sigma\in \mathcal{P}ath$, let $$W(\sigma,\varepsilon)=\{\gamma\in\mathcal{P}ath\colon d(\sigma(\tau),\gamma(\tau))<\varepsilon \text{  for all $\tau\in[0,1]$}\}$$
and $$\widetilde{W}(\sigma,\varepsilon)=\{X\in H^1_{\sigma}(TM):\|X\|_\infty=\sup_{\tau\in[0,1]}\sqrt{g(X(\sigma;\tau),X(\sigma;\tau))}<\varepsilon\}.$$ The exponential mapping 
 $\mathrm{exp}$  on the manifold $M$  
 generates a homeomorphism
$
\mathrm{Exp}_\sigma\colon  \widetilde{W}(\sigma,\varepsilon)\to W(\sigma,\varepsilon) 
$
by the formula
$$
\mathrm{Exp}_\sigma(X)(\tau)=\mathrm{exp}_{\sigma(\tau)}(X(\tau)).
$$ 
In fact, the structure of the Hilbert manifold on $\mathcal{P}ath$ can be defined by the atlas  $(\mathrm{Exp}_\sigma^{-1}, W(\sigma,\varepsilon))$,  $\sigma\in C^1([0,1],M)$.
The infinite dimensional Cristoffel symbols  $\Gamma^{\sim}$ of the connection $\nabla^{\mathcal{H}^0}$  in the coordinate chart  $(\mathrm{Exp}_\sigma^{-1}, W(\sigma,\varepsilon))$   are defined as follows.
  If
$\gamma\in W(\sigma,\varepsilon)$,
 $X\in H^0_{\gamma}(TM)$ and $Y\in H^1_{\gamma}(TM)$, then for almost all $\tau$ the vector
 $(\Gamma^{\sim}(\gamma)(X,Y))(\tau)\in T_{\gamma(\tau)}M$  has the following expression 
\begin{equation*}
(\Gamma^{\sim}(\gamma)(X,Y))(\tau)=\Gamma(\gamma(\tau))(X(\tau),Y(\tau))
\end{equation*}
in the normal coordinate chart on $M$ at the point
$\sigma(\tau)$.
 Then, in this coordinate chart on $M$ at the point
$\sigma(\tau)$ we have the following expression for the covariant derivative
\begin{equation}
\label{covdev}
\nabla^{\mathcal
H^0}_YX(\gamma;\tau)=d_YX(\gamma;\tau)+\Gamma(\gamma(\tau))(X(\gamma;\tau),Y(\gamma;\tau)),
\end{equation}
where $X\in C^\infty(W(\sigma,\varepsilon),\mathcal H^0)$ and $Y\in C^\infty(W(\sigma,\varepsilon),\mathcal H^1)$ are smooth local sections in $\mathcal H^0$ and $\mathcal H^1$ respectively.

The Hilbert manifold $\mathcal{P}ath_m=\{\gamma\in \mathcal{P}ath\colon \gamma(0)=m\}$ of  parametrized $H^1$ curves with origin at $m\in M$, the manifold $\mathcal{L}oop=\{\gamma\in \mathcal{P}ath\colon \gamma(0)=\gamma(1)\}$ of parametrized loops, the manifold $\mathcal{L}oop_{m}=\{\gamma\in \mathcal{P}ath\colon \gamma(0)=\gamma(1)=m\}$  of parametrized loops based at $m\in M$  are submanifolds of finite codimension of $\mathcal{P}ath$. By analogy, the bundles $\mathcal{H}^0(\mathcal{M})$ (the fiber over $\gamma\in \mathcal{M}$ is the space of  $L_2$-fields along curves along $\gamma$) and $\mathcal{H}^1_{0,0}(\mathcal{M})$ (the fiber over $\gamma\in \mathcal{M}$ is the space $(H^1_{0,0})_\gamma(TM)$) over $\mathcal{M}\in \{\mathcal{P}ath_m,\mathcal{L}oop_{m},\mathcal{L}oop\}$ can be defined.
Note that $\mathcal{H}^{1}_{0,0}(\mathcal{L}oop_{m})$ is a tangent bundle of the manifold $\mathcal{L}oop_{m}$.

\section{$H^0$-gradient,  L\'evy divergence and L\'evy Laplacian}
\label{SecCovLevyLaplacian}

In this section, we consider the L\'evy Laplacian and some related differential operators, and study some of their properties.

The L\'evy Laplacian on a manifold can be defined as a composition of special gradient and divergence. 
First, define the gradient on $C^\infty(\mathcal{P}ath,\mathbb{R})$ using the Riemannian metric $G_0$ on the bundle of $H^0$ vector fields.

\begin{definition}
\label{H_0grad}
The global smooth section $\mathrm{grad}_{H^0}\varphi$ in $\mathcal{H}^0$ is the $H^0$-gradient of the function $\varphi \in C^\infty(\mathcal{P}ath,\mathbb{R})$   if
$$
G_0(\mathrm{grad}_{H^0}\varphi(\gamma), X(\gamma))=d_{X} \varphi(\gamma)
$$
for any $\gamma\in \mathcal{P}ath$ and any local smooth section $X$ in $\mathcal{H}_{0,0}^1$.
\end{definition}
For examples see Sec.~\ref{Secfunctionals}.
Note that an arbitrary function from $C^\infty(\mathcal{P}ath,\mathbb{R})$ may not have an $H^0$-gradient. 
Direct calculations can verify the fulfillment of the Leibniz rule for the $H^0$-gradient.
\begin{proposition}
If $\varphi_1,\varphi_2\in C^\infty(\mathcal{P}ath,\mathbb{R})$
have $H^0$-gradients then 
$$
\mathrm{grad}_{H^0}(\varphi_1\varphi_2)=\varphi_2\mathrm{grad}_{H^0}\varphi_1+\varphi_1\mathrm{grad}_{H^0}\varphi_2.
$$
\end{proposition}

Let $T^2(\mathcal{H}^1_{0,0})$ be the vector bundle over $\mathcal{P}ath$  which fiber over  $\gamma\in \mathcal{P}ath$ is the space of continuous bilinear functionals  on ${(H^1_{0,0})}_\gamma(TM)$. 
Let $\hat{\otimes}^2T^*M$  and $\wedge^2T^*M$ be the bundles of
symmetic and antisymmetic tensors of type $(0,2)$ over $M$
respectively. Let  $\mathcal R_1$ be  the vector bundle over
$\mathcal{P}ath$ which fiber over  $ \gamma\in \mathcal{P}ath$ is the space of all
$H^0$-sections in $\hat{\otimes}^2T^*M$ along $\gamma$. Let
$\mathcal R_2$ be  the vector bundle over  $\mathcal{P}ath$ which fiber over
$ \gamma\in \mathcal{P}ath$ is the space of all $H^1$-sections in
$\wedge^2T^*M$ along $\gamma$.
Let $C^{\infty}_{AGV}(\mathcal{P}ath,T^2(\mathcal{H}^1_{0,0}))$  denote the space of all sections $K$
in $T^2(\mathcal{H}^1_{0,0})$ 
that have the form
\begin{multline}
\label{AGVtensors}
 K(\gamma)\langle X,Y\rangle=\\=\int_0^1\int_0^1
K^V(\gamma;\tau_1,\tau_2)\langle X(\gamma;\tau_1),Y(\gamma;\tau_2)\rangle d\tau_1d\tau_2+\int_0^1 K^L(\gamma;\tau)\langle X(\gamma;\tau),Y(\gamma;\tau)\rangle d\tau+\\
 +\frac 12\int_0^1 K^S(\gamma;\tau)\langle\nabla
X(\gamma;\tau), Y(\gamma;\tau)\rangle d\tau+\frac 12\int_0^1 K^S(\gamma;\tau)\langle\nabla
Y(\gamma;\tau), X(\gamma;\tau)\rangle d\tau,
\end{multline}
where $K^L$ and
$K^S$ are smooth sections in $\mathcal R_1$ and $\mathcal R_2$ respectively,
$K^V$ is a smooth section in the Hilbert tensor product bundle $\mathcal H^0\otimes \mathcal H^0$.
The kernel $K^V$ is called the Volterra kernel, $K^L$ is called the L\'{e}vy kernel and $K^S$ is called the singular kernel~\cite{AGV1993,AGV1994}.

\begin{definition}
The domain $\mathrm{dom}\, \mathrm{div}_L$ of the  L\'evy divergence consists of all $\psi\in C^\infty(\mathcal{P}ath,\mathcal{H}^0)$
such that there exists $K_\psi\in C^{\infty}_{AGV}(\mathcal{P}ath,T^2(\mathcal{H}^1_{0,0}))$  that
the following holds
\begin{equation*}
G_0(\nabla^{\mathcal
H^0}_{X}\psi(\gamma), Y(\gamma))= K_\psi(\gamma)\langle X(\gamma),Y(\gamma)\rangle
\end{equation*}
for any $\gamma\in \mathcal{P}ath$ and for any   local sections  $X,Y$ in $\mathcal{H}^1_{0,0}$.
 The L\'evy divergence is a linear mapping $\mathrm{div}_L\colon \mathrm{dom}\, \mathrm{div}_L \to  C^\infty(\mathcal{P}ath,\mathbb{R})$
defined by the formula
$$
\mathrm{div}_L \psi(\gamma)=\int_0^1  \operatorname{tr}_g( K^L_\psi(\gamma;\tau))d\tau,
$$
where $ \operatorname{tr}_g (K^L)=g^{\mu\nu}K^L_{\mu\nu}$ is the metric contraction of the L\'evy kernel $K^L_\psi$ of $K_\psi$.
\end{definition}

\begin{remark}
This definition is a modification of the definition of the functional divergence introduced in~\cite{AV1981}. In this form, the L\'evy divergence is introduced in~\cite{Volkov2019a}.  For the study of the connection of the L\'evy divergence with the Yang--Mills fields for the flat case and for the case of the Malliavin calculus see~\cite{Volkov2017,Volkov2019}. 
\end{remark}

\begin{proposition}
\label{Propdivergence}
If $\varphi$ has a $H^0$-gradient and $\psi$ belongs to the domain of the L\'evy divergence then
\begin{equation*}
\mathrm{div}_L (\varphi\psi)=\varphi\mathrm{div}_L \psi.
\end{equation*}
\end{proposition}
\begin{proof}
Direct calculations show that
$G_0(\nabla^{\mathcal
H^0}_{X}\psi, Y)= K_{\varphi\psi}\langle X,Y\rangle$,
where $$K^V_{\varphi\psi}\langle X,Y\rangle=\varphi K^V_{\psi}\langle X,Y\rangle+G_0(\mathrm{grad}_{H_0}\varphi,X)\psi\langle Y\rangle,$$
$
K^L_{\varphi\psi}=\varphi K^L_{\psi}$ and
$
K^S_{\varphi\psi}=\varphi K^S_{\psi}.
$
\end{proof}

\begin{definition}
\label{LevyLapl}
  The value of the L\'evy Laplacian $\Delta_L$ on $\varphi\in C^\infty(\mathcal{P}ath,\mathbb{R})$
is defined by
\begin{equation*}
\label{covLaplaceL}
\Delta_L \varphi=\mathrm{div}_L (\mathrm{grad}_{H^0} \varphi).
\end{equation*}
\end{definition}
For examples see Sec.~\ref{Secfunctionals}. 
\begin{proposition}
\label{ChainRule}
Let  $
\Phi=(\Phi_1,\ldots,\Phi_{N})\colon \mathcal{P}ath \to\mathbb{R}^{N}$   be 
a  smooth vector function such that all it components 
belong to the domain of the L\'evy Laplacian. Let $\mathcal{F}\in C^{\infty}(\mathbb{R}^N,\mathbb{R}) $ and $F=\mathcal{F}\circ\Phi$. Then
\begin{equation}
\label{DeltaF}
\Delta_LF(\gamma)=(\nabla \mathcal{F}(\Phi(\gamma)),\Delta_L\Phi(\gamma))_{\mathbb{R}^{N}}.
\end{equation}
\end{proposition}
\begin{proof}
By the chain rule 
$
\mathrm{grad}_{H^0} F(\gamma)=\sum_{k=1}^N\frac{\partial }{\partial x_k}\mathcal{F}(\Phi(\gamma))\mathrm{grad}_{H^0}\Phi_k(\gamma)$.
All functions $\frac{\partial }{\partial x_k}\mathcal{F}(\Phi(\cdot))$ belong to the domain of $H^0$-gradient and equality~(\ref{DeltaF}) follows from Proposition~\ref{Propdivergence}.
\end{proof}

As a corollary, we can obtain that the Leibniz rule is satisfied for the L\'evy Laplacian, as for a first-order differential operator. This property has been noted in the literature for the flat case (see~\cite{Accardi}).
 
\begin{proposition}
\label{Leibnizproductrule}
If $\varphi_1,\varphi_2\in C^\infty(\mathcal{P}ath,\mathbb{R})$ belong to the domain of the L\'evy Laplacian then 
$$
\Delta_L(\varphi_1\varphi_2)=\varphi_1\Delta_L\varphi_2+\varphi_2\Delta_L\varphi_1.
$$
\end{proposition}

The definitions of the $H^0$-gradient, the L\'evy divergence, and the L\'evy Laplacian carry over unchanged to manifolds $\mathcal{M}\in \{\mathcal{P}ath_m,\mathcal{L}oop_{m},\mathcal{L}oop\}$. In the text of the article we consider operators on such manifolds, using the same notations. Note that a function on $\mathcal{L}oop_{m}$ has an $H^0$-gradient if it is strongly differentiable in the sense of the work~\cite{Driver}.

\section{L\'evy trace and L\'evy Laplacian}
\label{Sec4}
In this section, we consider another definition of the L\'evy Laplacian on a manifold and prove its equivalence to the previous one.

The manifold $\mathcal{P}ath_m$  is parallelizable. Therefore, one can define the L\'evy Laplacian by analogy with the flat case, as a composition of a special linear functional and the second derivative.
This definition of the L\'evy Laplacian  coincides with the definition of the L\'evy Laplacian on manifolds from~\cite{LV2001}. For the flat case $M=\mathbb{R}^d$ this definition coincides with the definition from~\cite{AGV1994}.

Let us fix an orthonormal basis $\{Z_1,\ldots,Z_d\}$   in $T_mM$.
We identify the Hilbert spaces $H_0^1=H_0^1([0,1],\mathbb{R}^d)$ and $H_0^1([0,1],T_mM)$ by
$$
H^1_0([0,1],\mathbb{R}^d)\ni h(\cdot)=(h^\mu(\cdot))\leftrightarrow Z_\mu h^\mu(\cdot)\in H^1_0([0,1],T_mM).
$$
Due to~(\ref{Xgamma}), for any $\gamma\in \mathcal{P}ath_m$ the Levi-Civita connection generates the canonical isometrical isomorphism  between
$H^1_0$ and $H^1_\gamma(TM)$, which action on $h\in H^1_0$ we denote by $\widetilde{h}$. This isomorphism acts   by the formula
\begin{equation}
\label{isomorhism}
\widetilde{h}(\gamma;\tau)=Q_{\tau,0}(\gamma)h(\tau)=Z_\mu(\gamma,\tau)h^\mu(\tau),
\end{equation}
where $Z_\mu(\gamma,\tau)=Q_{\tau,0}(\gamma)Z_\mu$ for $\mu\in \{1,\ldots,d\}$. 
 Due to isomorphism~(\ref{isomorhism}), the tangent bundle over $\mathcal{P}ath_m$ is trivial. It implies that for any smooth function $\varphi\in C^\infty( \mathcal{P}ath_m,\mathbb{R})$ there exists the function $\widetilde{D}\varphi\in C^\infty(\mathcal{P}ath_m,H^1_{0,0})$ such that $d_{\widetilde{h}}\varphi(\gamma)=\langle\widetilde{D}\varphi(\gamma),h\rangle$ for any $h\in H^1_{0,0}$. Also there exists the function  $\widetilde{D}^2\varphi\in C^\infty(\mathcal{P}ath_m,L(H^1_{0,0},H^1_{0,0}))$ such that 
$\langle d_{\widetilde{h_2}}\widetilde{D}\varphi(\gamma),h_1\rangle=\langle\widetilde{D}^2\varphi(\gamma)h_2,h_1\rangle$ for any $h_1,h_2\in H^1_{0,0}$. Here $L(H^1_{0,0},H^1_{0,0})$ is the space of linear continuous operators on $H^1_{0,0}$.

Let $T^2_{AGV}$  be the space of all continuous bilinear real-valued functionals on    $H_{0,0}^1\times H_{0,0}^1$ that have the form
\begin{multline*}
Q(u,v)=\int_0^1\int_0^1Q^V(\tau_1,\tau_2)\langle u(\tau_1),v(\tau_2)\rangle d\tau_1d\tau_2+\int_0^1Q^L(\tau)\langle u(\tau),v(\tau)\rangle d\tau+
\\
+\frac 12\int_0^1Q^S(\tau)\langle \dot{u}(\tau),v(\tau)\rangle d\tau+\frac 12\int_0^1Q^S(\tau)\langle\dot{v}(\tau),u(\tau)\rangle d\tau,
\end{multline*}
where
$Q^V\in L_2([0,1]\times[0,1],T^2(\mathbb{R}^d))$,
$Q^L\in L_2([0,1],Sym^2(\mathbb{R}^d))$,
$Q^S\in H^1([0,1],\Lambda^2(\mathbb{R}^d))$.

\begin{definition}
The L\'evy trace $\mathrm{tr}^{AGV}_L$ acts on $Q\in T^2_{AGV}$ by the formula
$$
\mathrm{tr}^{AGV}_LQ=\int_0^1\mathrm{tr}\,Q^L(\tau)d\tau.
$$
\end{definition}

\begin{definition}
The domain of the L\'evy Laplacian $\Delta_L^{AGV}$  is the space of  all smooth functions  $\varphi$ on $\mathcal{P}ath_m$ such that $\widetilde{D}^2\varphi(\gamma)\in T^2_{AGV}$ for all $\gamma\in \mathcal{P}ath_m$. The L\'evy Laplacian $\Delta_L^{AGV}$ 
 acts
on $\varphi$ by the formula
$$
\Delta_L^{AGV}\varphi(\gamma)=
\mathrm{tr}^{AGV}_L(\widetilde{D}^2\varphi(\gamma)).
$$
\end{definition}

\begin{theorem} If the function $\varphi\colon\mathcal{P}ath_m\to \mathbb{R}$ belongs to the domain of the L\'evy Laplacian $\Delta_L$ then
$$
\Delta_L^{AGV}\varphi=\Delta_L\varphi.
$$
\end{theorem}
\begin{proof}
Due to $\varphi$ has the $H^0$-gradient,
$$
d_{\widetilde{h}}\varphi(\gamma)=\langle\widetilde{D}\varphi(\gamma),h\rangle
=G_0(\mathrm{grad}_{H^0}\varphi,\widetilde{h}).
$$
Hence
\begin{multline*}
\langle\widetilde{D}^2\varphi(\gamma)h_2,h_1\rangle=\langle d_{\widetilde{h_1}}\widetilde{D}\varphi(\gamma),h_2\rangle=d_{\widetilde{h}_1}G_0(\mathrm{grad}_{H^0}\varphi(\gamma),\widetilde{h}_2(\gamma))=\\=G_0(\nabla^{\mathcal
H^0}_{\widetilde{h}_1}\mathrm{grad}_{H^0}\varphi(\gamma),\widetilde{h}_2(\gamma))+G_0(\mathrm{grad}_{H^0}\varphi(\gamma),\nabla^{\mathcal
H^0}_{\widetilde{h}_1}\widetilde{h}_2(\gamma)).
\end{multline*}
The last equality holds due to the connection  $\nabla^{\mathcal
H^0}$ is a Riemannian connection.
Formula for the first derivative  the parallel transport in the case of the parallel transport generated by the Levi-Civita connection (see~\cite{Aref'eva1980,Driver,
Volkov2020}) implies 
\begin{multline}
\label{dhh}
d_{\widetilde{h}_1}(\widetilde{h}_2(\gamma;\tau_2))=d_{\widetilde{h}_1}Q_{\tau_2,0}(\gamma)h_2(\tau_2)=\\
=-\int_{0}^{\tau_2}Q_{\tau_2,\tau_1}(\gamma)R(\gamma(\tau_1))\langle Q_{\tau_1,0}(\gamma)h_2(\tau_2),\widetilde{h_1}(\gamma;\tau_1),\dot{\gamma}(\tau_1)\rangle d\tau_1-\\-\Gamma(\gamma(\tau_2))\langle\widetilde{h}_2(\gamma;\tau_2),\widetilde{h}_1(\gamma;\tau_2)\rangle,
\end{multline}
where $R=(R^\lambda_{\mu \nu \kappa})$ is a Riemannian curvature tensor.  Let $$\widetilde{R}_{\tau_2,\tau_1}(\gamma)\colon T_{\gamma(0)}M\times T_{\gamma(0)}M\to T_{\gamma(\tau_2)}M,$$ where $0\leq \tau_1\leq \tau_2\leq 1$, is defined as
$$\widetilde{R}_{\tau_2,\tau_1}(\gamma)\langle X,Y\rangle=Q_{\tau_2,\tau_1}(\gamma)R(\gamma(\tau_1))\langle Q_{\tau_1,0}(\gamma)X,Q_{\tau_1,0}(\gamma)Y,\dot{\gamma}(\tau_1)\rangle.$$
Formulas~(\ref{dhh}) and~(\ref{covdev}) together imply
$$
\nabla^{\mathcal
H^0}_{\widetilde{h}_1}\widetilde{h}_2(\gamma;\tau_2)=-\int_{0}^{\tau_2}\widetilde{R}_{\tau_2,\tau_1}(\gamma)\langle h_2(\tau_2),h_1(\tau_1)\rangle d\tau_1.
$$
Let $G_0(\nabla^{\mathcal
H^0}_{\widetilde{h}_1}\mathrm{grad}_{H^0}\varphi(\gamma), \widetilde{h}_2)= K(\gamma)\langle \widetilde{h}_1(\gamma),\widetilde{h}_2(\gamma)\rangle$. Then $\langle\widetilde{D}^2\varphi(\gamma)h_1,h_2\rangle=Q(\gamma)\langle h_1,h_2\rangle$,
where the Volterra kernel of $Q$ has the form
\begin{multline*}
\label{KV}
Q^V(\gamma;\tau_1,\tau_2)\langle h_1(\tau_1),h_2(\tau_2)\rangle=\\=\begin{cases}
K^V(\gamma;\tau_1,\tau_2)\langle\widetilde{h}_1(\gamma,\tau_1),\widetilde{h}_2(\gamma,\tau_2)\rangle-\\-g(\mathrm{grad}_{H^0}\varphi(\gamma;\tau_2),\widetilde{R}_{\tau_2,\tau_1}(\gamma)\langle h_2(\tau_2),h_1(\tau_1)\rangle)
,&\text{if $\tau_2\geq \tau_1$,}\\
K^V(\gamma;\tau_1,\tau_2)\langle\widetilde{h}_1(\gamma,\tau_1),\widetilde{h}_2(\gamma,\tau_2)\rangle
,&\text{if $\tau_2<\tau_1$},
\end{cases}
\end{multline*}
the L\'evy kernel of $Q$ has the form
\begin{equation*}
Q^L(\gamma;\tau)\langle h_1(\tau),h_2(\tau)\rangle=K^L(\gamma;\tau)\langle \widetilde{h}_1(\gamma;\tau),\widetilde{h}_2(\gamma;\tau)\rangle,
\end{equation*}
and the singular kernel of $Q$ has the form
\begin{equation*}
Q^S(\gamma;\tau)\langle \dot{h}_1(\tau),h_2(\tau)\rangle=K^S(\gamma;\tau)\langle \nabla\widetilde{h}_1(\gamma;\tau),\widetilde{h}_2(\gamma;\tau)\rangle.
\end{equation*}
Then
$$
\Delta_L^{AGV}\varphi(\gamma)=\int_0^1\mathrm{tr}\,Q^L(\gamma;\tau)d\tau=\int_0^1  \operatorname{tr}_g (K^L(\gamma;\tau))d\tau=\Delta_L\varphi(\gamma).
$$
\end{proof}

\section{L\'evy Laplacian as the Ces\`aro mean of the second order directional derivatives}
\label{CesaroMean}
In this section we consider a definition of the L\'evy Laplacian as the Ces\`aro mean of the second order directional derivatives and prove its equivalence to the previous two definitions.

The definition of the L\'evy Laplacian on the manifold  as the Ces\`aro mean of the second order directional derivatives was introduced in~\cite{AS2006} by analogy with the original P.~L\'evy's definition~(\ref{Levy2}). This definition depends on the choice of a special basis in $L_2([0,1],\mathbb{R})$. 
 An orthonormal basis $\{e_n\}$ in $L_2([0,1],\mathbb{R})$
is called  weakly uniformly dense~\cite{L1951} if 
$\lim_{n\to\infty}\int_0^1 h(\tau)\left(\frac 1n\sum_{k=1}^n
e_k^2(\tau)-1\right)d\tau=0$
for any $h\in L_{\infty}([0,1],\mathbb{R})$.
In this section, we assume that 
$\{e_n\}$ is a weakly uniformly dense  orthonormal basis in $L_2([0,1],\mathbb{R})$ such that all its elements belong to
  $H^1_{0,0}([0,1],\mathbb{R})$. We also assume that $\{e_n\}$ is uniformly bounded functions on the interval $[0,1]$. The main example of such an orthonormal basis is $e_n(\tau)=\sqrt{2}\sin{(n\pi \tau)}$. Let $e_{\mu, n}=Z_\mu e_n$, where $\{Z_\mu\}$ is an orthonormal basis in $\mathbb{R}^d$.  Then the L\'evy trace can be expressed as the Ces\`aro mean.
\begin{proposition}
\label{p}
If $Q\in T^2_{AGV}$, 
then
\begin{equation*}
\label{tr=tr}
\mathrm{tr}^{AGV}_L
Q=\lim_{n\to\infty}\frac 1n\sum_{k=1}^n\sum_{\mu=1}^dQ(e_{\mu,n},e_{\mu,n}).
\end{equation*}
\end{proposition}
For proof see~\cite{Volkov2019}.

As in Sec.~\ref{Sec4}, we identify $\{Z_1,\ldots,Z_d\}$ with the basis in $T_mM$. Then the vector functions $\{e_{\mu, n}\}$ generate the vector fields $\{\widetilde{e_{\mu, n}}\}$ on $\mathcal{P}ath_m$.
Moreover,
$\{\widetilde{e_{\mu, n}}(\gamma)\}$ form an orthonormal basis in  the fiber of the  bundle $\mathcal H^0$ over  $\gamma\in \mathcal{P}ath_m$
which elements belong to the fiber of the  bundle $\mathcal H^1_{0,0}$.

\begin{definition}
\label{def1}  The  L\'evy Laplacian $\Delta^{\{e_{\mu,n}\}}_L$, generated by  $\{e_{\mu,n}\}$, is a linear mapping $$\Delta^{\{e_{\mu,n}\}}_L:\mathrm{dom}\Delta^{\{e_{\mu,n}\}}_L\to
C^\infty(\mathcal{P}ath_m,\mathbb{R})$$ defined by
\begin{equation*}
\label{formld}
 \Delta^{\{e_{\mu,n}\}}_L\varphi(\gamma)=\lim_{n\to\infty}\frac 1n\sum_{k=1}^n\sum_{\mu=1}^d\left.\frac {d^2}{ds^2}\right|_{s=0}\varphi\left(\mathrm{Exp}_\gamma(s\widetilde{e_{\mu,k}})\right),
\end{equation*}
where  $\mathrm{dom}  \Delta^{\{e_{\mu,n}\}}_L$ is the space of all functions $\varphi\in C^\infty(\mathcal{P}ath_m,\mathbb{R})$ such that the right side of~(\ref{formld}) is a function from  $C^\infty(\mathcal{P}ath_m,\mathbb{R})$.  
\end{definition}

\begin{theorem}
Let $\varphi$ belong to the domain of the operator $\Delta^{AGV}_L$. Then
$$
\Delta^{\{e_{\mu,n}\}}_L\varphi=\Delta_{L}^{AGV}\varphi.
$$
\end{theorem}
\begin{proof}
For   any $\gamma\in \mathcal{P}ath_m$  and any $h\in H^1_{0,0}$ hold
$
\frac {d}{ds}\varphi(\mathrm{Exp}_\gamma(s\widetilde{h}))=\langle\widetilde{D}\varphi(\mathrm{Exp}_\gamma(s\widetilde{h}),h\rangle
$ and
$$\frac {d^2}{ds^2}\varphi(\mathrm{Exp}_\gamma(s\widetilde{h}))=\langle d_{\widetilde{h}}\widetilde{D}\varphi(\mathrm{Exp}_\gamma(s\widetilde{h})),h\rangle=\langle\widetilde{D}^2\varphi(\mathrm{Exp}_\gamma(s\widetilde{h}))h,h\rangle.$$
Hence, $\left.\frac {d^2}{ds^2}\right|_{s=0}\varphi(\mathrm{Exp}_\gamma(s\widetilde{h}))=\langle\widetilde{D}^2\varphi(\gamma)h,h\rangle.$
Then Proposition~\ref{p} implies that
$$
\Delta^{\{e_{\mu,n}\}}_L\varphi(\gamma)=\lim_{n\to\infty}\frac 1n\sum_{k=1}^n\sum_{\mu=1}^d\langle\widetilde{D}^2\varphi(\gamma)e_{\mu,k},e_{\mu,k}\rangle=\mathrm{tr}^{AGV}_L(\widetilde{D}^2\varphi(\gamma))=\Delta_{L}^{AGV}\varphi(\gamma).
$$
\end{proof}
\section{ L\'evy Laplacian and differential forms}
\label{Secfunctionals}
In this section, using eigenforms for the Hodge-de Rham Laplacian we construct some eigenfunctions of the L\'evy Laplacian operator.

Bellow $M$ is a compact orientable Riemannian manifold without boundary.
Let $\ast\colon \Omega^p(M)\to \Omega^{d-p}(M)$ be the Hodge star on the manifold $M$. Let $\mathrm{d}\colon \Omega^p(M)\to \Omega^{p+1}(M)$ be the operator of the exterior derivative and  $\delta\colon \Omega^{p+1}(M)\to \Omega^{p}(M)$ be the codifferential. So 
$\delta = (-1)^{d(p + 1) + 1} \ast \mathrm{d} \ast$.
Let
$$
\Delta=- \mathrm{d}\delta-\delta\mathrm{d} =-(\mathrm{d}+\delta)^2
$$
be the Hodge–de Rham Laplacian (we choose the sign of the operator so that it is negative semidefinite on a compact manifold, despite the fact that 
in the literature on geometry the opposite sign is often chosen~\cite{Jost}). For $p=0$ this Laplacian coincides with the Laplace-Beltrami operator. A $p$-form $\omega$ is harmonic if $\Delta \omega=0$. Let $\mathcal{H}_\Delta^p(M) = \{\omega\in\Omega^p(M)\mid\Delta\omega=0\}$ be the space of harmonic $p$-forms.

Let $\mathfrak{f}\in C^\infty(M,\mathbb{R})$. Let
$\mathfrak{L_{f}}\colon \mathcal{P}ath\to \mathbb{\mathbb{R}}$
be defined by
\begin{equation*}
\mathfrak{L_{f}}(\gamma)=\int_0^1\mathfrak{f}(\gamma(\tau))d\tau.
\end{equation*}
\begin{proposition}
\label{forfunct}
It holds  that
\begin{equation}
\label{laplLf}
 \Delta_L\mathfrak{L_{f}}(\gamma)=\int_0^1 \Delta \mathfrak{f}(\gamma(\tau))d\tau,
\end{equation}
where $\Delta$ is the Laplace-Beltrami operator on the manifold $M$. 
\end{proposition}
\begin{proof}
The $H^0$-gradient of the function $\mathfrak{L_{f}}$ has the form 
$$
\mathrm{grad}_{H^0}\mathfrak{L_{f}}(\gamma;\tau)=\operatorname{grad}\mathfrak{f}(\gamma(\tau)),
$$
where $\operatorname{grad}$ is the gradient on the Riemannian  manifold $M$. By direct calculations one can obtain
\begin{multline}
\label{L_f}
G_0(\nabla^{\mathcal
H^0}_{X}\mathrm{grad}_{H^0}\mathfrak{L_{f}}(\gamma), Y(\gamma))=\\
=\int_0^1g(\nabla_{X(\gamma,\tau)} \operatorname{grad} f(\gamma(\tau)),Y(\gamma,\tau))d\tau
=\int_0^1
\operatorname{Hess}f(\gamma(\tau))(X(\gamma,\tau), Y(\gamma,\tau))d\tau,
\end{multline}
where $\operatorname{Hess}$ is the Hessian tensor.
The functional $\mathfrak{L_{f}}$ belongs to the domain of the L\'evy Laplacian $\Delta_L$ and the  Volterra part and the singular part in~(\ref{L_f}) vanish. This implies~(\ref{laplLf}).
\end{proof}

If $a\in \Omega^1(M)$, 
let the functional $\Theta_a\colon \mathcal{P}ath\to \mathbb{R}$ be defined by
$$
\Theta_a(\gamma)=\int_{\gamma}a=\int_0^1a(\gamma(\tau))\dot{\gamma}(\tau)d\tau.
$$
Note  that  if $b\in \Omega^1(M)$ and $a-b$ is an exact 1-form then 
$
\Theta_a(\gamma)=\Theta_b(\gamma)
$
for any $\gamma\in \mathcal{L}oop$. If $a$ is a closed form, then $\Theta_a$ takes the same value on homotopy equivalent curves from $\mathcal{L}oop$.
Let the functional $U^a\colon  \mathcal{P}ath\to \mathbb{C}$ be defined by
$$
U^a(\gamma)=e^{-i\Theta_a(\gamma)}.
$$
The 1-form $ia$ can be interpreted as  a $U(1)$-connection in a trivial bundle over $M$ and the functional $U^a$ can be  interpreted as a parallel transport generated by this connection. 
\begin{proposition}
\label{Phaseprop1}
If $a\in \Omega^1(M)$ then the functional  $\Theta_a$ on  $\mathcal{P}ath$ belong to the domain of the L\'evy Laplacian and
\begin{equation}
\label{LLofPhase}
\Delta_L\Theta_a(\gamma)=-\int_{\gamma}\delta \mathrm{d}a.
\end{equation}
On the loop space $\mathcal{L}oop$ it holds
\begin{equation}
\label{LoopLLofPhase}
\Delta_L\Theta_a(\gamma)=\int_{\gamma}\Delta a.
\end{equation}
\end{proposition}
\begin{proof}
Let
$f=\mathrm{d}a$, i.e. $f=\sum_{\mu<\nu}f_{\mu\nu}dx^\mu\wedge dx^\nu$, where $f_{\mu\nu}=\partial_\mu a_\nu-\partial_\nu a_\mu$.
By direct calculations and by integrating by parts, taking into account $X(\gamma,0)=X(\gamma,1)=0$, we can obtain
$$
G_0(\mathrm{grad}_{H^0}\Theta_a(\gamma),X(\gamma)
)=\int_0^1f(\gamma(\tau))\langle X(\gamma,\tau), \dot{\gamma}(\tau)\rangle d\tau.
$$
Also, using direct calculations, integration by parts and the  Bianchi identities, one can obtain (cf.~\cite{Volkov2019a})
\begin{multline}
\label{ThetaVolterr}
G_0(\nabla^{\mathcal
H^0}_{X}\mathrm{grad}_{H^0}\Theta_a(\gamma),Y(\gamma))
=\frac 12\int_0^1\nabla f(\gamma(\tau))\langle
X(\gamma;\tau), Y(\gamma;\tau),\dot{\gamma}(\tau)\rangle d\tau+
\\+\frac 12\int_0^1\nabla f(\gamma(\tau))\langle
Y(\gamma;\tau),X(\gamma;\tau),\dot{\gamma}(\tau)\rangle d\tau-\\
-\frac 12\int_0^1f(\gamma(\tau))\langle\nabla
X(\gamma,\tau), Y(\gamma;\tau)\rangle d\tau-\frac 12\int_0^1 f(\gamma(\tau))\langle \nabla
Y(\gamma;\tau), X(\gamma;\tau)\rangle d\tau,
\end{multline}
where $\nabla f$ is a covariant derivative of $f$ due to the Levi-Civita connection on the manifold $M$. Then $\Theta_a$ belongs to the domain of the L\'evy Laplacian. Note that the Volterra kernel in~(\ref{ThetaVolterr}) vanishes.
In local coordinates,
$$
(\delta f)_\nu=-\nabla^\mu f_{\mu\nu}=-g^{\lambda \mu}(\partial_\lambda f_{\mu \nu}-f_{\mu
\kappa}\Gamma^\kappa_{\lambda\nu}-f_{\kappa\nu}\Gamma^\kappa_{\lambda\mu}).
$$
This implies~(\ref{LLofPhase}).
Due to $\mathrm{d}\delta a$ is an exact 1-form, on the loop manifold $\mathcal{L}oop$ we have
\begin{equation*}
\Delta_L\Theta_a(\gamma)=-\int_{\gamma}\delta \mathrm{d} a-\underbrace{\int_{\gamma}\mathrm{d} \delta a}_{=0}=\int_\gamma \Delta a.
\end{equation*}
\end{proof}

\begin{proposition}
The functional  $U^a$ on  $\mathcal{P}ath$ belongs to the domain of the L\'evy Laplacian
\begin{equation*}
\label{transport_Laplacian1}
\Delta_LU^a(\gamma)=-iU^a(\gamma)\int_{\gamma}\delta \mathrm{d}a.
\end{equation*}
On the loop space $\mathcal{L}oop$ it holds
\begin{equation*}
\label{transport_Laplacian}
\Delta_LU^a(\gamma)=ie^{-i\Theta_a(\gamma)}\int_\gamma \Delta a.
\end{equation*}
\end{proposition}
\begin{proof}
Follows directly from Proposition~\ref{ChainRule}  and Proposition~\ref{Phaseprop1}.
\end{proof}

Following the observation from~\cite{Accardi}, we
 can construct a family of eigenfunctions for the L\'evy Laplacian. For eigenvalues and eigenvectors of
 the Hodge–de Rham Laplacian on a compact Riemannian manifold
see e.g.~\cite{Chavel}.

\begin{proposition}
Let smooth functions $\mathfrak{f}_1,\ldots\mathfrak{f}_{N_1}$ on $M$ be  eigenfunctions of the Laplace-Beltrami operator corresponding to eigenvalues $\lambda_1,\ldots,\lambda_{N_1}$ respectively.
Let smooth 1-forms $a_1,\ldots,a_{N_2}$ on $M$
 be  eigenforms of the Hodge–de Rham Laplacian  corresponding to eigenvalues $\mu_1,\ldots,\mu_{N_2}$ respectively. Then, the functional $F:\mathcal{L}oop \to \mathbb{R}$ defined by
$$
F(\gamma)=\left(\prod_{j=1}^{N_1}\mathfrak{L}_{\mathfrak{f}_j}(\gamma)\right)\left(\prod_{k=1}^{N_2}\Theta_{a_k}(\gamma)\right)
$$
is an eigenfunction of the L\'evy Laplacian corresponding to the eigenvalue 
$(\sum_{j=1}^{N_1}\lambda_j+\sum_{k=1}^{N_2}\mu_k)$.
\end{proposition}
\begin{proof}
It follows directly from the Leibniz product rule  for the L\'evy Laplacian (Proposition~\ref{Leibnizproductrule}), Proposition~\ref{forfunct} and formula~(\ref{LoopLLofPhase}) from Proposition~\ref{Phaseprop1}.
\end{proof}

\section{Heat equations with  L\'evy Laplacian and with Hodge-de Rham Laplacian}
\label{SecHeat}

In this section, it is shown that the heat flows of differential 0-forms and 1-forms generate solutions of the heat equation with the L\'evy Laplacian. These solutions tend to locally constant functionals for the L\'evy Laplacian as time tends to infinity.

Classical Hodge theorem states that there exists $L_2$-decomposition for $\Omega^p(M)$ ($p\in \{0,\ldots, d\}$)
$$
\Omega^p(M)=\mathrm{d}\Omega^{p-1}(M)\oplus\mathcal{H}_\Delta^p(M)\oplus\delta \Omega^{p+1}(M). 
$$
So any $\omega\in \Omega^p(M)$ can be decomposed
$$
\omega=\mathrm{d}\alpha+h+\delta\beta,
$$
where $\alpha\in \Omega^{p-1}(M)$, $\beta\in \Omega^{p+1}(M)$ and $h\in \mathcal{H}_\Delta^p(M)$. The
differential forms $\mathrm{d}\alpha$, $h$, and $\delta\beta$ are uniquely defined.

Consider the initial value problem for the heat equation for differential $p$-form $\omega=\omega(t)=\omega(t,x)$ (we assume $\omega$ is $C^{\infty}$ with respect to $(t,x)$):
\begin{equation}
\label{1_form_heat}
\begin{cases}
\frac{\partial \omega}{\partial t}=\Delta \omega;\\
\omega(0)=\omega_0\in \Omega^p(M).
\end{cases}
\end{equation}
Milgram--Rosenbloom~
theorem~\cite{MilgramRosenbloom,Chavel,Jost} states that the  unique solution of~(\ref{1_form_heat}) exists for all $t>0$ and converges in $C^\infty$ on $M$ to the harmonic form $h_0$ as $t$ tends to infinity. This harmonic form $h_0$ is an orthogonal projection of $\omega_0$ on $\mathcal{H}_\Delta^p(M)$.

\begin{theorem}
Let  $\mathfrak{f}_1,\ldots\mathfrak{f}_{N_1}$  and  $a_1,\ldots,a_{N_2}$ be smooth solutions of the heat equations for $0$- and $1$-forms on $M$ respectively. 
Let  $\Phi=(\Phi_1,\ldots,\Phi_{N})\colon [0,+\infty)\times \mathcal{L}oop\to\mathbb{R}^{N}$, where $N=N_1+N_2$, such that
\begin{equation}
\label{functmain}
\Phi_j(t,\gamma)=\begin{cases}
\mathfrak{L_{f_j}}(t,\gamma)=\int_0^1\mathfrak{f}_j(t,\gamma(\tau))d\tau
,&\text{if $j\in\{1,\ldots,N_1\}$,}\\
\Theta_{a_{(j-N_1)}}(t,\gamma)=\int_{\gamma}a_{(j-N_1)}(t)
,&\text{if $j\in\{N_1+1,\ldots,N\}$}.
\end{cases}
\end{equation}
Let $\mathcal{F}\in C^{\infty}(\mathbb{R}^N,\mathbb{R})$ and $F(t,\gamma)=\mathcal{F}(\Phi(t,\gamma))$. 
Then $F(\cdot,\cdot)$ is a solution of the heat equation for the L\'evy Laplacian
$
\partial_t F=\Delta_LF.
$
This solution $F(t,\cdot)$ pointwise converges to a locally constant functional on $\mathcal{L}oop$ as  $t$ tends to infinity.  This functional takes the same value on all homotopy equivalent curves.
\end{theorem}
\begin{proof}
If $\mathfrak{f}=\mathfrak{f}(t,x)$ is a solution of the heat equation for the Laplace-Beltrami operator
$\partial_t\mathfrak{f}=\Delta \mathfrak{f}$, then the time-depended family of functionals   $\mathfrak{L_{f}}(t,\gamma)=\int_0^1\mathfrak{f}(t,\gamma(\tau))d\tau$ satisfies
$$\partial_t \mathfrak{L_{f}}(\gamma)=\int_0^1\partial_t\mathfrak{f}(t,\gamma(\tau))d\tau=\int_0^1\Delta\mathfrak{f}(t,\gamma(\tau))d\tau=\Delta_L\mathfrak{L_{f}}(\gamma)$$
on the manifold $\mathcal{L}oop$.
Due to all harmonic functions on the compact manifold are constant, Milgram--Rosenbloom theorem implies that $\mathfrak{L_{f}}(t,\cdot)$ pointwise converges to a  constant functional on $\mathcal{L}oop$ as $t\to \infty$. Similarly, if  $a=a_\mu(t,x)dx^\mu$ is a solution of the heat equation $
\partial_t a=\Delta a$  for 1-forms on $M$, then 
the time-depended family of functionals 
$\Theta_a(t,\gamma)=\int_{\gamma}a(t)
$  satisfies 
$$\partial_t\Theta_a(t,\gamma)=\int_{\gamma} \partial_t a(t)=\int_{\gamma}\Delta a(t)=\Delta_L\Theta_a(t,\gamma).$$
It follows directly  from Milgram--Rosenbloom theorem that $\Theta_a(t,\cdot)$ converges pointwise to a functional given by the integral of the harmonic 1-form on $M$ as  $t\to \infty$. This integral functional is constant on homotopy equivalent paths. Hence, 
Proposition~\ref{ChainRule} implies
$$
\partial_t F(t,\gamma)=(\nabla \mathcal{F}(\Phi(t,\gamma)),\partial_t\Phi(t,\gamma))_{\mathbb{R}^{N}}=
(\nabla \mathcal{F}(\Phi(t,\gamma)),\Delta_L\Phi(t,\gamma))_{\mathbb{R}^{N}}=\Delta_LF(t,\gamma)
$$
and $F(t,\cdot)$ pointwise converges to a locally constant functional on $\mathcal{L}oop$ as $t\to \infty$.
\end{proof}

The particular case of the previous theorem is the following. If  $a=a_\mu(t,x)dx^\mu$ is a solution of the heat equation $
\partial_t a=\Delta a$  for 1-forms on $M$, then 
the time-depended family of the parallel transports
$
U^a(t,\gamma)=e^{-i\Theta_a(t,\gamma)}
$   satisfies
the heat equation for the L\'evy Laplacian
$\partial_tU^a=\Delta_LU^a$ on the manifold $\mathcal{L}oop$.

\begin{remark}
The heat equation for the L\'evy Laplacian has been studied extensively (see e.g.~\cite{Feller2005,K2003}). However, the known results on the existence and uniqueness classes for this equation are not applicable to functionals~(\ref{functmain}) defining by 1-forms. The study of the heat equation with the L\'evy Laplacian remains an open problem, of interest due to its connection with the Yang-Mills heat equation, which is discussed briefly in the next section.
\end{remark}

\section{Yang--Mills heat flow}
\label{SecYMH}

The connection between the linear heat equation with the L\'evy Laplacian for parallel transport $U^a$ and the heat flow for 1-form $a$ is preserved when passing from the commutative case to the non-commutative one. Namely, the equation for parallel transport with the L\'evy Laplacian becomes equivalent to the Yang-Mills heat equations for a time-depended connection, which are quasi-linear. 
For a detailed exposition of the theory of Yang-Mills gradient flows see e.g.~\cite{Feehan2024,Waldron}.

 Let
$E=E(\mathbb{R}^N,\pi,M,G)$  be a vector bundle over $M$ with the
projection $\pi\colon E\to M$ and the structure group  $G\subseteq
\operatorname{SO}(N)$.  Let  the Lie algebra of the  structure group
be $\mathfrak{Lie}(G)\subseteq \mathfrak{so}(N)$.
 The fiber over $x\in M$ is $E_x=\pi^{-1}(x)\cong\mathbb{R}^N$.
 Let $P$ be the principle bundle  over
$M$  associated with $E$ and  $\mathrm{ad} (P)=\mathfrak{Lie}(G)\times_G M$   be the
adjoint bundle of $P$ (the fiber of $\mathrm{ad}P$ is
  isomorphic to $\mathfrak{Lie}(G)$). A connection $A(x)=A_\mu(x)dx^\mu$ in the vector bundle $E$ is  a smooth section in
$\Lambda^1(T^\ast M)\otimes \mathrm{ad}P$. Locally, $A$
  is a smooth $\mathfrak{Lie}(G)$-valued 1-form. 
The curvature $F$ of the connection $A$ is a smooth section in $\Lambda^2(T^\ast M)\otimes \mathrm{ad}P$.
In the local trivialization, the curvature $F$ is the $\mathfrak{Lie}(G)$-valued 2-form,
where  $F_{\mu\nu}=\partial_\mu A_\nu-\partial_\nu A_\mu+[A_\mu,A_\nu]$.
The gauge transform is a smooth section in $\mathrm{Aut} P$. Such a section $\psi$ acts on the connection $A$ by the formula
\begin{equation*}
A\to A'=\psi^{-1}A\psi+\psi^{-1}d\psi
\end{equation*}
and on the curvature $F$ by the formula
\begin{equation*}
F\to F'=\psi^{-1}F\psi.
\end{equation*}

The Yang--Mills action functional is defined as $L^2$-norm of the curvature 
$$
S_{YM}(A)=\frac 12\int_{M}\|F(x)\|^2\mathrm{Vol}(dx),
$$
where  $\|F(x)\|$ is  the natural norm on $\mathfrak{Lie}(G)$-valued 2-forms generated by the Killing form on 
$\mathfrak{Lie}(G)$  and $\mathrm{Vol}$ is the Riemannian  measure on $M$.  The Yang--Mills equations are
the
Euler--Lagrange equations for this action functional. These equations  are
\begin{equation*}
D^\ast F=0,
\end{equation*}
where $D^\ast$ is a covariant сodifferential.
Locally, 
$(D^\ast F)_\nu=-\nabla^\mu F_{\mu\nu}$,
where
\begin{equation*}
 \nabla_\lambda F_{\mu
\nu}=\partial_\lambda F_{\mu \nu}+[A_\lambda,F_{\mu
\nu}]-F_{\mu
\kappa}\Gamma^\kappa_{\lambda\nu}-F_{\kappa\nu}\Gamma^\kappa_{\lambda\mu}.
\end{equation*}
The Yang--Mills heat equations are nonlinear parabolic differential
equations on a time-depended connection $A$ of the form
\begin{equation}
\label{ymh}
\partial_t A=-D^\ast F.
\end{equation}
The Yang--Mills heat flow was introduced as the gradient flow for the Yang--Mills  action functional in~\cite{AtBo}. 
In~\cite{Donaldson,DK} such a  flow was studied for the case of stable holomorphic bundles over  compact Kahler surfaces for the Hermitian Yang--Mills equations.
The existence of a Yang--Mills heat flow for all times and its convergence to a Yang--Mills connection on Riemannian manifolds of dimensions two and three was proved in~\cite{Rade}. For the critical dimension 4, the short time existence was proved in~\cite{Struwe}. The existence for all times was initially proved for the flow of the equivariant connections~\cite{SchlatterStruweTavildah-Zadeh}  and later in the general case~\cite{Waldron2016,
Waldron}. In case of dimension greater than or equal to 5, blow-up occurs and the solution cannot be continued to an infinite interval
~\cite{Naito}.

Let  $\mathcal{E}_m$ be a Hilbert vector bundle over the Hilbert manifold $\mathcal{P}ath_m$ such that its fiber over  $\gamma\in \mathcal{P}ath_m$ is the space  $\mathrm{Hom}(E_m,E_{\gamma(1)})$. The definitions of the  L\'evy Laplacian  $\Delta_L$  can been transferred  to the space of sections in the bundle $\mathcal{E}_m$. Let $U^A(\gamma)$ denote the parallel transport along the curve $\gamma\in \mathcal{P}ath_m$ generated by the connection $A$ in the vector bundle $E$.
The mapping $\mathcal{P}ath_m\ni \gamma\mapsto U^A(\gamma)$ is a $C^{\infty}$-smooth section in the vector bundle $\mathcal E_m$ (for the proof of smoothness see~\cite{Gross,Driver}). Let us formulate the theorem on the relationship between the classical solution of the Yang-Mills heat equation and the heat equation for the L\'evy Laplacian.

\begin{theorem}
\label{maintheorem}
Let $A(\cdot,\cdot)\in C^{\infty}([0,T]\times M,\Lambda^1\otimes adP)$.
The following two assertions are equivalent:
\begin{enumerate}
\item  the flow of connections $[0,T]\ni t\mapsto A(t)=A(t,\cdot)$ is a
solution of the Yang--Mills  heat equations~(\ref{ymh});
\item the flow of parallel transports   $[0,T]\ni t\mapsto U(t)=U^{A(t)}$  is a solution
of the heat equation for the L\'evy Laplacian
$\partial_tU=\Delta_LU$.
\end{enumerate}
\end{theorem}
For proof see~\cite{Volkov2019a}. 
Note that, if $A(\cdot,\cdot)\in C^{\infty}([0,\infty)\times M,\Lambda^1\otimes adP)$ satisfies~(\ref{ymh}) and converges in $C^\infty$ to a Yang-Mills connection  as $t\to \infty$ then the flow of the parallel transports  $U(t)$ pointwise  tends to the solution of the Laplace equation
for the L\'evy Laplacian. We conjecture  that there should be a stronger convergence of the solutions of the L\'evy Laplacian heat equation. A natural question is the behavior of the solution of the heat equation with the L\'evy Laplacian  in the case of the blow-up of the solution of the Yang-Mills heat equations. 
\begin{remark}
To prove the equivalence in Theorem~\ref{maintheorem}, it is essentially used that the infinite dimensional base manifold is the manifold of paths with a free end and not the manifold of loops. For the manifold of loops, it can be shown that the second condition of the theorem follows from the first condition.
\end{remark}

\section{Conclusion}
\label{Sec:Conclusions}
 
In this article, we systematized different definitions of the L\'evy Laplacian on an infinite-dimensional manifold of paths  and proved the equivalence of these definitions. It is shown that the eigenfunctions of the Laplace-Beltrami operator and the Hodge-de Rham Laplacian for 1-forms generate a family of eigenfunctions for the L\'evy Laplacian. It is shown that the heat flows of differential 0-forms and 1-forms generate solutions of the heat equation with the L\'evy Laplacian. These solutions tend to locally constant  functionals as time tends to infinity. The relationship between the solutions of the  heat equation with the L\'evy Laplacian and the Yang-Mills heat equations is discussed.

\section*{Acknowledgments}
This work was supported by the Russian Science Foundation under grant no. 24-11-00039, https://rscf.ru/en/project/24-11-00039/ .

\bibliographystyle{unsrt}
\bibliography{references}
\end{document}